\documentclass[conference,letterpaper]{IEEEtran}

\usepackage{graphicx}

\usepackage[utf8]{inputenc} 
\usepackage[T1]{fontenc}
\usepackage{url}
\usepackage{ifthen}
\usepackage{cite}
\usepackage[cmex10]{amsmath} 

\usepackage{multirow}
\usepackage{amssymb}
\usepackage{amsfonts} 

\usepackage{amssymb}
\usepackage{mathtools}
\usepackage{amsthm}

\newtheorem{theorem}{Theorem}
\newtheorem{proposition}[theorem]{Proposition}
\newtheorem{lemma}[theorem]{Lemma}
\newtheorem{corollary}[theorem]{Corollary}
\newtheorem{remark}[theorem]{Remark}

\usepackage{color, colortbl}
\usepackage{mathrsfs}

\newcommand{\T}{\top}
\newcommand{\st}{\mathrm{s.t.}}

\newcommand{\argmax}{\mathop{\mathrm{argmax}}}

\newcommand{\EE}{ \mathbb{E} } 
\newcommand{\R}{ \mathbb{R} }

\newcommand{\EEb}[1]{ \EE \left[ {#1} \right]}

\newcommand{\tr}[1]{ \mathrm{Tr}\left\{ {#1} \right\} }

\newcommand{\hlf}{\frac{1}{2}}

\usepackage{enumerate}

\def\wrt{w.r.t\onedot}

\newcommand{\X}{ \mathcal{X} }
\newcommand{\Y}{ \mathcal{Y} }
\newcommand{\bfP}{{\bf P}}
\newcommand{\bfQ}{{\bf Q}}
\newcommand{\bfD}{{\bf D}}
\newcommand{\bfH}{{\bf H}}
\newcommand{\bfA}{{\bf{A}}}

\newcommand{\bfPi}{{\bf{\Pi}}}
\newcommand{\bfrho}{{\bf{\rho}}}

\def\WLOG{\emph{w.l.o.g.}\onedot}
\def\hadm{\odot}
\def\onedot{~}

\newcommand{\EEQ}[1]{ \EE_{\bfQ} \left[ {#1} \right]}

\usepackage{balance}

\begin{document}
\title{Characterization of the Distortion-Perception Tradeoff for Finite Channels with Arbitrary Metrics} 

\author{%
 \IEEEauthorblockN{Dror Freirich, Nir Weinberger and Ron Meir}
 \IEEEauthorblockA{Viterbi Faculty of Electrical and Computer Engineering\\
                   Technion - Israel Institute of Technology                 }
}

\maketitle

\begin{abstract}
   Whenever inspected by humans, reconstructed signals should not be distinguished from real ones. Typically, such a high perceptual quality comes at the price of high reconstruction error, and vice versa.
We study this distortion-perception (DP) tradeoff over finite-alphabet channels, for the Wasserstein-$1$ distance induced by a general metric as the perception index, and an arbitrary distortion matrix. Under this setting, we show that computing the DP function and the optimal
reconstructions is equivalent to solving a set of linear programming problems. We provide a structural characterization of the DP tradeoff, where the DP function is piecewise linear in the perception index. We further derive a closed-form expression for the case of binary sources.
\end{abstract}

\section{Introduction}
\begin{figure*}[t]
\center
\includegraphics[viewport=8bp 8bp 420bp 305bp,clip,width=.38\linewidth]{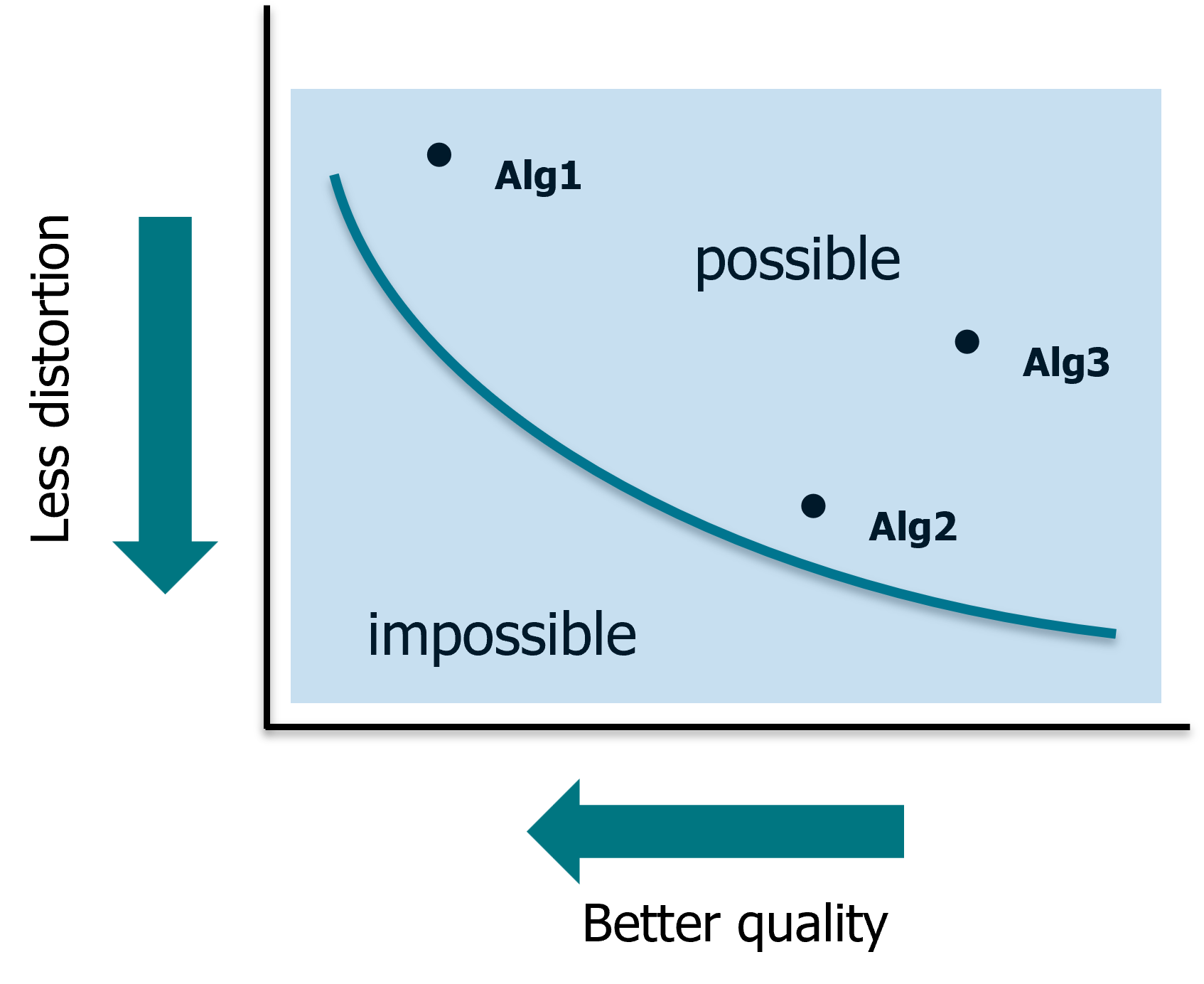}
\includegraphics[viewport=20bp  22bp 420bp 305bp,clip,width=.38\linewidth]{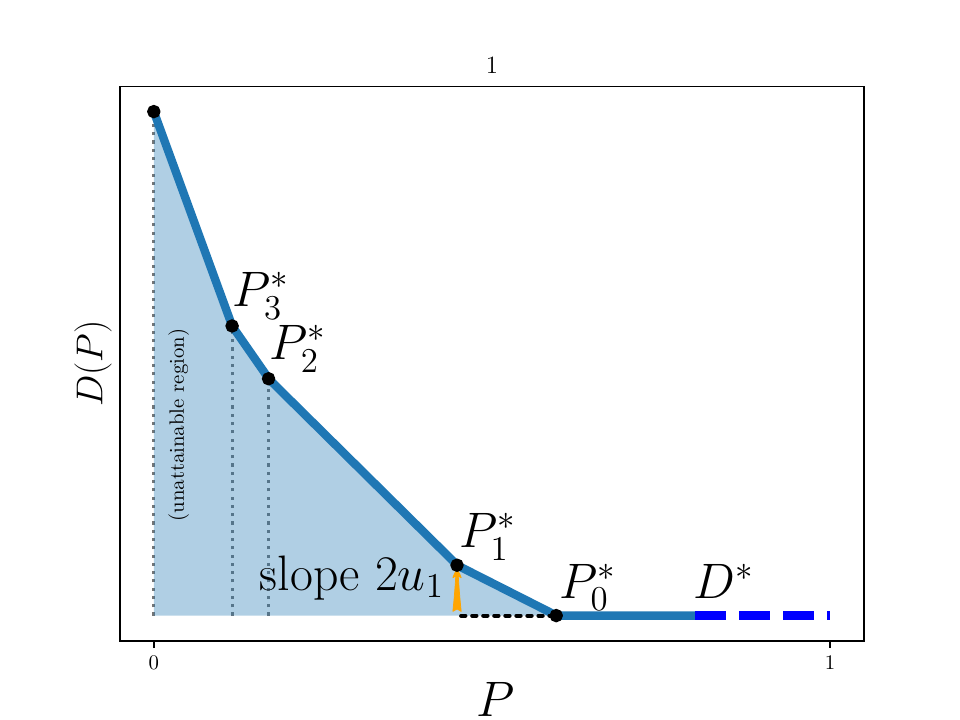}
\caption{\textbf{The distortion-perception (DP) function.} \label{fig:dp}
\textbf{(Left)} The minimal distortion possible for a certain level of perceptual quality forms a convex, non-increasing  curve. The region below the curve can not be attained by any reconstruction method. \textbf{(Right)} In our discrete setting,
$D(P)$ is a piecewise linear function. Breakpoints $P^*_i$ and slopes $2u_i$ are given explicitly by Theorem~\ref{thm::DP:binarysource} for binary sources.}
\end{figure*}
The reconstruction of a signal from degraded data is required in numerous settings across science and engineering. 
Until recently, reconstruction algorithms' performance has been measured by its mean \textit{distortion}, such as mean squared error (MSE). For that reason, many methods aimed to minnimize distortion measures such as MSE and peak signal-to-noise ratio (PSNR). 
However, in systems whose outputs are inspected by  human users,
reconstructions should not be easily distinguished from signals typical to the source domain.
Therefore, many current works target \textit{perceptual quality}
rather than distortion (\textit{e.g.} in image restoration, see  \cite{adrai2023deep, wang2018esrgan,lim2017enhanced,ledig2017photo}).

Mathematically, the probability of success in a hypothesis test is known to be proportional to the Total-Variation (TV) distance between distributions \cite{nielsen2013hypothesis}.
Hence, high {perceptual quality} is considered to be achieved when the distribution of restored signals is close to the real signals distribution \cite{blau2018perception}.

Good {perceptual
quality} generally comes at the price of high reconstruction error and vice versa. 
This leads to a tradeoff between distortion and perception, first studied in \cite{blau2018perception}. The central problem is thus to quantify the \textit{distortion-perception (DP) function}, which is the minimal distortion possible for a certain level of perceptual quality. 
The DP problem was studied by various authors. Specifically, \cite{freirich2021theory} studied the DP function in real spaces, for the MSE distortion and the Wasserstein-$2$ perception index. In discrete spaces, \cite{qian2022rate} characterized the special case of a binary source, for the Hamming distortion and the TV perception index.

In this paper, we focus on discrete spaces, and investigate the DP tradeoff for general finite-alphabet channels and  general distortion matrices. As the perception index, we consider the \textit{Wasserstein}-$1$ distance induced by a general metric, which generalizes the TV distance \cite{ambrosio2008gradient,van2014probability,raginsky2013concentration}. 
We show that finding the DP function and the optimal reconstruction for this setting is equivalent to solving a set of \textit{linear} problems, and the result is always a piecewise linear function of the perception index, regardless of the channel size, the underlying distributions or distortion measure. This stems from the properties of the \textit{dual} feasible set.
We further revisit the binary setting of \cite{qian2022rate}, and derive a closed-form expression for the DP function, now considering a general distortion measure. We provide a self-contained proof for this case based on our novel analysis of the general setting.

\section{Preliminaries}
\subsection{The distortion-perception tradeoff}
\label{appsec::intro}

Let $X,Y$ be random variables taking values in some complete separable metric spaces $\mathcal{X\mathrm{,}Y}$, respectively. We assume the existence of the joint probability $p_{X,Y}$ on $\mathcal{X}\times\mathcal{Y}$, and a Borel lower-bounded distortion
	function $d:\mathcal{X}\times\mathcal{X}\rightarrow \R^{+}\cup\{0\}$.
 An estimator $\hat{X}\in\mathcal{X}$ is a random variable on $\X$, defined by its
 distribution conditioned on the measurement $Y$, $p_{\hat{X}|Y}$, with marginal distribution $p_{\hat{X}}$. 
 
An \textit{optimal} estimator for the DP tradeoff, is an estimator that minimizes the expected \textit{distortion} $\EE [d(X,\hat{X})]$ under the \textit{perception} constraint $d_{p}(p_X,p_{\hat{X}})\leq P$. Here, $d_{p}$ is a divergence between probability measures.
\cite{blau2018perception} introduced the DP \textit{function}
	\begin{equation}
	D(P)\triangleq\min_{p_{\hat{X}|Y}}\left\{\EE[d(X,\hat{X})] \;:\; d_{p}(p_X,p_{\hat{X}})\leq P\right\}.
 \label{APP_eq:D_P::General_definition}
	\end{equation}
	The expectation is taken \textit{w.r.t.} the probability measure induced by $p_{XY}$ and $p_{\hat X|Y}$ where we assume  that $X,\hat{X}$ are independent given $Y$.
We have the following result of \cite[Thm.~2]{blau2018perception}.
\begin{theorem}
\label{appthm::the_DP_Blau}
    (The perception-distortion tradeoff).
    If $d_p(p, q)$ is convex in its
    second argument, then the distortion-perception function \eqref{APP_eq:D_P::General_definition} is monotonically non-increasing and convex.
\end{theorem} 

\noindent This is the case for the TV and Wasserstein distances discussed in this paper.

\subsection{Related work}
Apart from the general properties given by Theorem \ref{appthm::the_DP_Blau}, the precise nature of DP functions depends on the exact setup. \cite{freirich2021theory} fully characterized this function in real spaces, considering the MSE and Wasserstein-$2$ indices. 
Under this setting, the distortion-perception function is always quadratic, and possesses a closed-form expression for Gaussian channels.
In this work we discuss discrete signals, where we provide an analogous structural characterization, in which $D(P)$ is always piecewise linear (Theorem~\ref{thm:piecewiseLinear}).

Reconstruction problems with constrained output distributions were studied in optimal transport \cite{10017280}, lossy compression and quantization \cite{6963478,7138587}.
Recently, \cite{freirich2023perceptual} investigated the cost of perfect perceptual consistency constraints in online estimation settings.
The DP tradeoff was also extended to lossy compression by presenting the rate-distortion-perception (RDP) function \cite{blau2019rethinking,salehkalaibar2023rate,yan2022optimally}, which is the minimal rate of a code whose decoding allows a desired tradeoff between reconstruction and perceptual quality. A coding theorem was introduced for this setting \cite{theis2021coding,wagner2022rate}, where the properties of optimal codes are investigated \cite{chen2022rate}.

In the context of RDP theory, \cite{qian2022rate} investigated channels with binary sources. 
They showed \cite[Thm.~7]{qian2022rate} that for the Hamming loss with the TV perceptual index, the DP function is a piecewise linear function whose breakpoints are given by an explicit formula.
Here, we extend this result, considering an arbitrary distortion measure (Theorem~\ref{thm::DP:binarysource}).

\section{Problem formulation\label{subsec:Discrete::Problem-formulation}}

In this paper, we discuss the discrete case, where $\X$ and $\Y$ are finite spaces. Let $X,Y$ be discrete variables defined on finite alphabets
$\mathcal{X}=\{x_1,\ldots,x_{n_x}\}, \mathcal{Y}=\{y_1,\ldots,y_{n_y}\}$,
where $X$ is the variable of interest, and $Y$ is a measurement of $X$ over a noisy channel. Their joint probability $p_{X,Y}\in\mathcal{P}(\mathcal{X}\times\mathcal{Y})$
is represented by the matrix $\bfP_{X,Y}=\{p(x,y)\}_{x,y\in\mathcal{X}\mathcal{\times Y}} \in \mathbb{R}^{|\mathcal{X}|\times|\mathcal{Y}|}$, and the marginal distributions $p_X$ and $p_Y$ are given by the vectors $\bfP_{X}\in\R^{|\mathcal{X}|},\bfP_{Y}\in\R^{|\mathcal{Y}|}$.
We assume that for each letter in
the channel's output, $p_Y(y_{i})>0$ (\textit{i.e.}, we ignore unused symbols). 
A randomized estimator $\hat{X}\in\mathcal{X}$ of $X$ from $Y$ is defined by a stochastic
transition matrix $\bfQ=\bfQ_{\hat{X}|Y}\in\R^{|\mathcal{X}|\times|\mathcal{Y}|}$
whose entries are the probabilities $q(\hat{x}|y)$ to reconstruct
the symbol $\hat{x}\in\mathcal{X}$ given that the channel output is $Y=y\in\mathcal{Y}$. We assume the
Markov relation where $X,\hat{X}$
are independent given $Y$.
The arbitrary \textit{distortion} matrix is given by $\bfD=\{d(x,\hat{x})\}_{x,\hat{x}\in\mathcal{X}^{2}}\in\R^{|\mathcal{X}|\times|\mathcal{X}|}$,
where the expected distortion 
\begin{equation}
\EEQ{d(X,\hat{X})}=\tr{\bfP_{X,Y}^{\T}\bfD\bfQ}
\end{equation}
 should be minimized w.r.t. $q(\hat{x}|y),\hat{x},y\in\mathcal{X}\times\mathcal{Y}$.
The marginal distribution $p_{\hat{X}}$ of $\hat{X}$ is given by
the vector $\bfP_{\hat{X}}=\bfQ\bfP_{Y}$. 
We are interested in analyzing the \textit{distortion-perception}
(DP) \textit{function} 
\cite{blau2018perception}
 \begin{equation}
D(P)=\min_{\bfQ_{\hat{X}|Y}}\left\{ \EEQ{d(X,\hat{X})}\;:\;d_{p}(p_{X},p_{\hat{X}})\leq P\right\}.
\end{equation}

For simplicity, let us first consider the TV distance as the perceptual index $d_p$, given by
\begin{flalign}
d_{TV}(\bfP_{X},\bfP_{\hat{X}}) &\triangleq\frac{1}{2}\sum_{x\in\mathcal{X}}|\bfP_{X}(x)-\bfP_{\hat{X}}(x)| \nonumber
\\
&=\sup_{A \subseteq \mathcal{X}}|p_{X}(A)-p_{\hat{X}}(A)|.
\label{eq:tv::def}
\end{flalign}
Note that using this definition, $d_{TV}(\bfP_{X},\bfP_{\hat{X}})\in[0,1]$, and $d_{TV}(\bfP_{X},\bfP_{\hat{X}})=0$ iff $\bfP_{X}=\bfP_{\hat{X}}$.
Now, 
\begin{flalign}
&D(P)=& \label{eq:DP::LinearOpt}
\\
\nonumber
&\min_{\bfQ\geq 0}\left\{ (\bfD^{\T}\bfP_{X,Y})\bullet \bfQ: \text{\ensuremath{\begin{array}{c}
\ensuremath{\boldsymbol{1}^{|\mathcal{X}|}\cdot \bfQ}=\boldsymbol{1}^{|\mathcal{Y}|}\\
d_{TV}(\bfP_{X},\bfQ\bfP_{Y})\leq P
\end{array}}} \!\! \right\} ,&
\end{flalign}
where the Frobenius inner product $A\bullet B=\tr{A^{\T}B}$,
$\boldsymbol{1}^{d}$ is the $1\times d$ dimensional all-ones vector, and 
for $\bfQ\in\mathbb{R}^{|\mathcal{X}|\times|\mathcal{Y}|}$ the constraint $\bfQ \geq 0$ is applied elementwise.
We start by presenting some elementary properties of \eqref{eq:DP::LinearOpt}.
\begin{proposition}
\label{prop::feasible}
Let $P\in[0,1]$. The optimization problem \eqref{eq:DP::LinearOpt}
is feasible (namely, the constraints are satisfiable), and its optimal value is bounded from below.
\end{proposition}
\begin{proof} 
The \textit{posterior sampling} solution $\bfQ=\bfP_{X|Y}=\{p_{X,Y}(x,y)/p_{Y}(y)\}_{x,y\in\mathcal{X}\mathcal{\times Y}}$
is feasible for every $P\geq 0$, since $\bfP_{\hat{X}}=\bfQ\bfP_{Y}=\bfP_{X}$, yielding $d_{TV}(\bfP_{\hat{X}},\bfP_{X})=0$. For every
stochastic matrix \bfQ,
\begin{equation}
(\bfD^{\T}\bfP_{X,Y})\bullet \bfQ\in \left[ \min D_{x,\hat{x}},\max D_{x,\hat{x}}\right],
\end{equation}
hence the optimal value is bounded.
\end{proof}
\begin{proposition}
\label{thm:Dstar}
Denote the matrix $\rho\triangleq \bfD^{\T}\bfP_{X,Y}$, whose entries are given by
$\rho_{\hat{x},y}=\bfP_Y(y)\EEb{d(X,\hat{x})|Y=y}$. 
Then, for any $P\geq 1$, 
\begin{equation}
D(P) = \sum_{y}\min_{\hat{x}\in\mathcal{X}}\rho_{\hat{x},y}\triangleq D^{*}.
\end{equation}
A corresponding optimal estimator is given by \begin{equation} \hat{X}^{*} (Y) \in\mathop{\mathrm{argmin}}_{\hat{x}}\rho_{\hat{x},Y}.
\end{equation}
Trivially, $D(P)\ge D^{*}$ holds for every $P\in[0,1]$.
\end{proposition}

\noindent The proof is straightforward.

\section{Linear Programming formulation}

We now observe that the perceptual constraint 
$\frac{1}{2} \sum_{x\in\mathcal{X}}|\bfP_{X}(x)-\sum_{y\in\mathcal{Y}}\bfP_{Y}(y)\bfQ(x|y)|\leq P$ 
in (\ref{eq:DP::LinearOpt}), is equivalent to the 
set of \textit{linear} constraints
\begin{equation}
\sum_{x\in\mathcal{X}}\pm\left(\bfP_{X}(x)-\sum_{y\in\mathcal{Y}}\bfP_{Y}(y)\bfQ(x|y)\right)\leq 2P.
\label{eq:dtv-2x-linear-constraints}
\end{equation}
Taking all possible sign combinations
we attain $2^{|\mathcal{X}|}$ linear constraints, where the $2$ constraints for which the signs 
 are either all positive or all negative are redundant since for probability vectors $\bfP_X,\bfP_Y$ and a stochastic matrix $\bfQ$ the LHS of \eqref{eq:dtv-2x-linear-constraints} vanishes. 
Together with \eqref{eq:DP::LinearOpt}, we can reformulate the DP function
as the following Linear Program (LP) \cite{bertsimas1997introduction,vanderbei2020linear}
\begin{flalign}
&D(P)=&\label{eq:DP::LPform} \\
&\min_{\bfQ \geq 0}\! \left\{  \rho \bullet \bfQ: \!\! \begin{array}{c}
\ensuremath{\boldsymbol{1}^{|\mathcal{X}|}\cdot  \bfQ}=\boldsymbol{1}^{|\mathcal{Y}|},  \  \bfQ \in \mathbb{R}^{|\mathcal{X}|\times|\mathcal{Y}|} \\
\sum_{x\in\mathcal{X}}\pm\left(\bfP_{X}(x)-\sum_{y\in\mathcal{Y}}\bfP_{Y}(y)\bfQ_{x|y} \right) \\ \leq 2P
\end{array} \!\! \right\}.&\nonumber
\end{flalign}
In (\ref{eq:DP::LPform}), we have $|\mathcal{X}|\times|\mathcal{Y}|$ variables (the entries of 
$\bfQ=\{q(\hat{x}|y)\}$), and 
 $|\mathcal{Y}|+2^{|\mathcal{X}|}-2$ constraints.
\subsection{Total Variation as a Wasserstein distance}
Let $\bfH=\{1-\delta_{x,\hat x}\}_{x,\hat x\in \X \times \X}$ be the \textit{Hamming} distance matrix, let $\mathscr{P}(\X)$ be the set of probability measures on $\X$, and let $\Pi\in\mathscr{P}(\X \times \X)$ be a coupling between $\bfP_{X}$ and $\bfP_{\hat{X}}$ (parameterized by a  matrix $\bf \Pi_{x, \hat x}$). It is well known \cite{van2014probability} that taking $\bfH$ as a metric on $\X$, the TV distance coincides with the\textit{ Wasserstein-$1$} distance on $\mathscr{P}(\X)$, namely
\begin{flalign}
d_{TV}(\bfP_{X},\bfP_{\hat{X}}) & = \inf_{\Pi}\Pi\left[x\neq\hat x\right] = W_{1,H}(\bfP_{X},\bfP_{\hat{X}}),&
\\
 W_{1,H}(\bfP_{X},\bfP_{\hat{X}}) &\triangleq \inf_{\bfPi}  \bfPi\bullet \bfH=\inf_{\bfPi} \sum_{x,\hat x} \bfPi_{x,\hat x} \bfH_{x,\hat x},&
\label{eq::TVasOT}
\end{flalign}
where the minimum is attained \cite[Lemma~3.4.1]{raginsky2013concentration}. Wasserstein distances are convex metrics on $\mathcal{P}(\X)$ \cite{ambrosio2008gradient}. Using \eqref{eq::TVasOT}, we can rewrite \eqref{eq:DP::LPform} as the linear problem
\begin{flalign}
&D(P)=& \label{eq:DP::OTform}
\\
\nonumber
&\min_{ \scriptsize \begin{array}{c}
\bfQ, \bfPi, \\ \varepsilon \geq 0
\end{array}
}\!\!\!\!
\left\{ \rho \bullet \bfQ:\!\!
\begin{array}{l}
\sum_{\hat x\in \X}\bfP_Y(y) \bfQ_{\hat x | y}=\bfP_Y(y), \forall y \in \Y\\
\sum_{\hat x\in \X}\bfPi_{x,\hat x }=\bfP_X(x), \forall x \in \X\\
\sum_{x\in \X}\bfPi_{x,\hat x }=\sum_{y \in \Y} \bfP_Y(y) \bfQ_{\hat x | y}, \forall \hat x \in \X
\\
\bfPi \bullet \bfH + \varepsilon =P,
\bfQ\in\mathbb{R}^{|\mathcal{X}|\times|\mathcal{Y}|},
\\
\bfPi \in\mathbb{R}^{|\mathcal{X}|\times|\mathcal{X}|}
\end{array}\!\!\!\!
\right\}& \!
\end{flalign}
where $\varepsilon$ is a slack variable.
 The problem \eqref{eq:DP::OTform} possesses $|\X |(|\Y| + |\X|)+1$ variables and only $|\Y| + 2|\X|+1$ constraints, from  which $|\Y| + 2|\X|$ are independent.

Interestingly, the form \eqref{eq:DP::OTform} allows to discuss a more general family of perceptual divergences $-$ Wasserstein-$1$ distances \eqref{eq::TVasOT} induced by arbitrary metrics $H$ on $\X$, which we will consider to be the case from this point on. We will assume \WLOG that $H$ takes values in $[0,1]$, hence the results of Propositions \ref{prop::feasible} and \ref{thm:Dstar} hold trivially in this case. 

\subsection{The Dual Problem}

Let the general linear programming problem \cite{bertsimas1997introduction}
\begin{equation}
{\rm (LP)} \quad
\min_q z^\T q, \,
\st  \ \bfA q= b ~\mathrm{and}~ q\ge 0 , 
\end{equation}
where
$q,z\in \mathbb{R}^n, b\in \mathbb{R}^{n_c}, \bfA\in \mathbb{R}^{n_c\times n}$,  
and the inequality is elementwise. Its dual problem (DLP) is given by
\begin{equation}
{\rm (DLP)} \quad
\max_w w^{\T}b, \, \st \  w^\T \bfA \leq z^\T.
\label{eq:dp::stndrd}
\end{equation}
Finally, recall that \textit{Strong duality} holds for feasible and bounded LP problems  \cite{bertsimas1997introduction}, namely, the problem \eqref{eq:dp::stndrd} is feasible and
\begin{equation}
\min_{q}z^\T q=\max_{w}w^{\T}b. \label{eq::DPdiscrete:strong duality}
\end{equation}
We next derive the dual form of \eqref{eq:DP::OTform}. For convenience, we split the variables in \eqref{eq:dp::stndrd}  into four groups: $|\mathcal{Y}|$ variables $\{w_y\}_{y\in \mathcal{Y}}$ related to the stochasticity constraint on $\bfQ$ for each symbol in $\mathcal{Y}$,  the two groups of ${|\mathcal{X}|}$ variables $\{r_x\}$ and $\{\nu_{\hat{x}}\}$ related to the constraints on the marginals of $\bfPi$, and the variable $l$ related to the perception constraint $\bfPi \bullet \bfH + \varepsilon =P$. We denote 
\begin{equation}
\bfrho'_{\hat{x},y} \triangleq \frac{\bfrho_{\hat{x},y}}{\bfP_Y({y})} =\EEb{d(X,\hat{x})|Y=y},
\label{eq:rho_prime}
\end{equation}
and explicitly write the dual problem of \eqref{eq:DP::OTform} as (see derivation in the Appendix),
\begin{flalign}
&\max_{w,r,\nu,l} \left[ \sum_{y\in \Y}p_{y}w_{y} + \sum_{x\in \X}p_{x}r_{x} -lP \right]&
\label{eq:DPdiscrete_explicit}
\\
\nonumber
&\st \, \left\{  l\geq 0,
\begin{array}{ll}
w_{y} \leq \rho_{\hat{x},y}'-\nu_{\hat{x}}, &
 \forall \hat{x},y\in\mathcal{X}\times\mathcal{Y}
 \\ 
 r_x \leq \bfH_{x, \hat x}l+\nu_{\hat x},\,&\forall x, \hat x \in \X \times \X 
\end{array}\right. .&
\end{flalign}
From the strong duality property, we have that \eqref{eq:DPdiscrete_explicit} is feasible; indeed, we can choose $w_y = \min_{\hat x \in \X} \rho'_{\hat x, y}$ and $r_x, \nu_{\hat x},l =0$. This choice of variables recovers the lower bound of Proposition~\ref{thm:Dstar}, where $D(P)\geq D^*$ for $P\in[0,1]$.

\begin{remark}
\label{rem::existence_vertices}
It is easy to see that in this case ${\rm rank}(\bfA) = |\Y| + 2|\X|$, while one constraint is redundant, namely we can eliminate a linear constraint from the primal program \eqref{eq:DP::OTform} (a row of $\bfA$) such that the \textit{row rank} of the problem is full. Equivalently, we can set one of the variables $r_x,\nu_{\hat x}$ to $0$, and the dual feasible set (projected onto $\R^{|\Y|+2|\X|}$) will not contain a line. This implies the existence of an extreme point in this dual set in $\R^{|\Y|+2|\X|}$ (see \cite[Thm. 2.6]{bertsimas1997introduction}).
\end{remark}

Given a value $P$, $D(P)$ can be calculated by numerically solving \eqref{eq:DP::OTform} (equivalently, \eqref{eq:DPdiscrete_explicit}). However,
finding a closed-form solution remains an open problem. In Section \ref{sec::binarychannel} we find such an expression for small alphabets.
We also observe that the objective of \eqref{eq:DPdiscrete_explicit}
 is linear in the perception index, hence the maximal value for a given $P$ is attained by some non-increasing linear function of the form $p_0+p_1P$. We further develop this insight below.

\section{Main results}
\subsection{Piecewise linearity of DP functions} 

\label{appsec::piecewise_linear}
\begin{figure*}[bt]
\center \includegraphics[viewport=75bp  40bp 1050bp 430bp, clip,width=.80\linewidth]{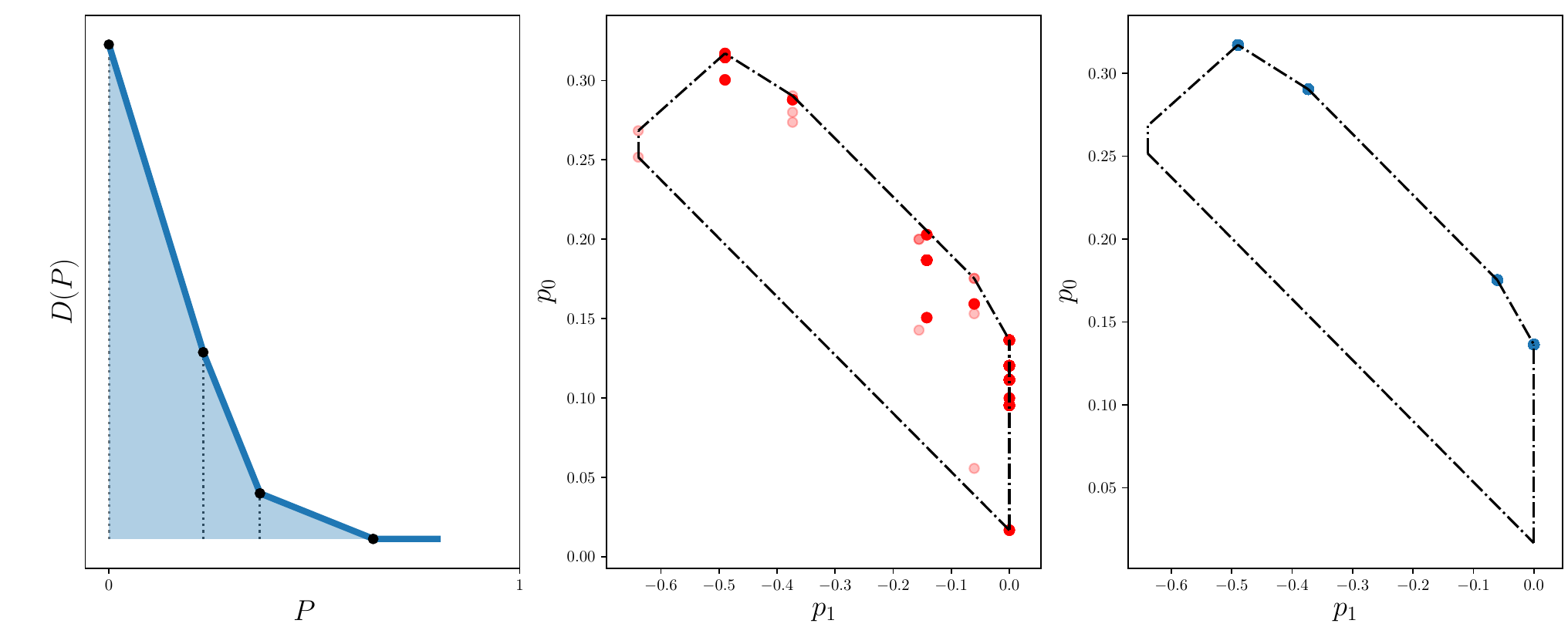}

\caption{\label{fig::convexhull}  {\bf Numerical illustration of Theorem~\ref{thm:piecewiseLinear} and Theorem~\ref{thm::r2convex}} for $|\X|=3, |\Y|=5$. In the {\bf (Middle)} pane we present the set $\mathcal{S}_2$ and its convex hull in the $(p_0,p_1)$-plane. The {\bf (Right)} pane shows the optimal solutions obtained by numerically solving \eqref{eq:DPdiscrete_explicit} for different values of $P$. We can see that the solutions, corresponding to the linear segments of $D(P)$ ({\bf Left} pane),  occur at extreme points of $\mathrm{conv}\left(\mathcal{S}_{2}\right)$.}
\end{figure*}
While the problem of finding an exact formula for $D(P)$ is still open, here we exploit the properties of the dual problem \eqref{eq:DPdiscrete_explicit}
 in order to show the general property that $D(P)$ is piecewise linear in the perception index $P$. 
 Moreover, the breakpoints and slopes of this function are determined by the vertices of a convex set in $\R^2$. 
 We will utilize the following  property of LP problems.
\begin{lemma}(\cite[Thm. 2.8]{bertsimas1997introduction})
\label{lemma:max_on_extreme}
For a bounded LP problem, if there exists an extreme point in the feasible set, then the optimal solution is obtained at an extreme point.  
\end{lemma}
\noindent This is true of course also for the dual problem. We now use this result to prove the following.

\begin{theorem}
\label{thm:piecewiseLinear}
For $P\in[0,\infty)$, the DP function 
\eqref{eq:DP::OTform}
is
a non-increasing piecewise linear function of $P$ with a non-decreasing
slope. Furthermore, there exists $P^{*}\in[0,1]$ such that $D(P)=D^{*},P\geq P^{*}$.
\end{theorem}

The proof is based on analyzing the dual formulation \eqref{eq:DPdiscrete_explicit}. Due to strong duality \eqref{eq::DPdiscrete:strong duality} this matches the primal problem. 
The feasible set of \eqref{eq:DPdiscrete_explicit} has a finite number of vertices, and this set is independent of the perceptual index $P$. The solution to \eqref{eq:DPdiscrete_explicit} must occur at one of these vertices. Thus, the interval $[0,1]$ may be partitioned into sub-intervals, so that in each sub-interval the solution to \eqref{eq:DPdiscrete_explicit} is at the same vertex. For a fixed choice of variables $w,r,\nu$ and $l$ in \eqref{eq:DPdiscrete_explicit}, the $D(P)$ function is linear with slope $-l$. Hence, the DP function is piecewise linear. Since DP functions are non-increasing and convex (see Thm.~\ref{appthm::the_DP_Blau}), the slope cannot decrease.

\begin{proof}[Proof of Theorem~\ref{thm:piecewiseLinear}]
Let $d=[0^{|\mathcal{Y}|+2|\X|},\ 1^{}]^{\T}$  and
$b_{0}=[\bfP_{Y}^{\T}, \bfP_{X}^{\T},0^{|\X| + 1}]$, both in $\R^{|\Y|+2|\X| +1}$. We can write the objective \eqref{eq:DPdiscrete_explicit}  
as \begin{equation}
[w,r,\nu,l]^{\T}b(P)\rightarrow\max_{w,r,\nu,l \in \mathcal{S} },\end{equation} where $b(P)=b_{0}-dP$ and $\mathcal{S}$ is the set of feasible solutions to the dual problem \eqref{eq:DPdiscrete_explicit} where we choose to set $\nu_{|\X|} \equiv 0$ (see Remark~\ref{rem::existence_vertices}). 
Let $\mathrm{ext}\left(\mathcal{S} \right) = \{p^{i}=[w^{i},r^{i},\nu^{i},l^{i}]\}$ denote
the vertices of $\mathcal{S}$ in $\R^{|\Y|+2|\X|}$. Note that the set of vertices is non-empty, finite, and importantly, independent of $P$. Lemma~\ref{lemma:max_on_extreme} above implies that the dual optimal value is obtained on this set. We now have from strong duality,
\begin{flalign}
D(P)&=\max_{i}p^{i}\cdot b(P) \nonumber \\ &=\max_{i}[w^{i},r^i,\nu^{i},l^i]^{\T}b(P) 
=\max_{i}\left[p_{0}^{i}+p_{1}^{i}P\right],
\label{eq::DPdiscrete:pcslinear}
\end{flalign}
where we denote the projections
\begin{align}
p_{0}^{i} & =p^{i}\cdot b_{0}, \label{eq::projrction0}\\
p_{1}^{i} & =-p^{i}\cdot d.
\label{eq::projrction1}\end{align}
As a maximum of finite set of linear functions, \eqref{eq::DPdiscrete:pcslinear}
 is a piecewise linear function. The non-decreasing slope property can be easily deduced from \eqref{eq::DPdiscrete:pcslinear}, or from the fact that DP functions are convex \cite{blau2018perception}.
\end{proof}

\begin{corollary}
The breakpoints of the $D(P)$ function lie within the set 
\begin{equation}
    \mathcal{P} =\left\{ {\frac{p_{0}^{i}-p_{0}^{j}}{p_{1}^{j}-p_{1}^{i}}:\,}   
    \begin{array}{c}
     p^i,p^j \textrm{ are vertices of the set of } \\ 
     \textrm{feasible
solutions to the dual problem}
\end{array}
\right\}.
\end{equation}
\end{corollary}

As we show next, not every vertex is a candidate for optimality in \eqref{eq::DPdiscrete:pcslinear}; optimal solutions must be obtained on a 2-D convex hull. 
Using the notations of Theorem~\ref{thm:piecewiseLinear} proof, denote the set $\mathcal{S}_{2}=\left\{ \left(p_{0}^{i},p_{1}^{i}\right):p_{0}^{i}=p^{i}\cdot b_{0},p_{1}^{i}=-p^{i}\cdot d,p^{i}\in \mathrm{ext}\left(\mathcal{S}\right)\right\} \subseteq\mathbb{R}^{2}$
which represents the (finite) set of linear curves $\left\{ p_{0}^{i}+p_{1}^{i}P, p^i \in \mathrm{ext}\left(\mathcal{S}\right) \right\} $ on the $2$-dimensional plane by the projections of their corresponding  vertices \eqref{eq::projrction0}-\eqref{eq::projrction1}.

\begin{theorem}
\label{thm::r2convex}
For any $P\geq 0$, there exists a vertex of $\mathcal{S}$ such that \textbf{$p^{k}\in\argmax_{p \in \mathrm{ext}\left(\mathcal{S}\right) }p^{}\cdot b(P)$},
and $\left(p_{0}^{k},p_{1}^{k}\right)$ is an extreme point of $\mathrm{conv}\left(\mathcal{S}_{2}\right)$.
\end{theorem}

\begin{proof}
Let \textbf{$\left\{ \left(\widetilde{p}_{0}^{k},\widetilde{p}_{1}^{k}\right)\right\} _{k=1}^{M}\subseteq\mathcal{S}_{2}$
}be the set of extremals of $\mathrm{conv}\left(\mathcal{S}_{2}\right)$. The set $\mathcal{S}_{2}$ is finite, hence its convex hull is bounded. We
can write any point in $\mathcal{S}_{2}$  as a convex combination
$\left(p_{0}^{i},p_{1}^{i}\right)=\sum_{k=1}^{M}\alpha_{ik}\left(\widetilde{p}_{0}^{k},\widetilde{p}_{1}^{k}\right)$,
thus we have 
\begin{flalign}
\nonumber p^{i}\cdot b(P)&=p_{0}^{i}+p_{1}^{i}P=\sum_{k=1}^{M}\alpha_{ik}\left(\widetilde{p}_{0}^{k}+\widetilde{p}_{1}^{k}P\right) &
\\
&\leq\max_{k}\left(\widetilde{p}_{0}^{k}+\widetilde{p}_{1}^{k}P\right)=\max_{k}{p}^{k}\cdot b(P).& 
\end{flalign}
\end{proof}

The results of Theorems~\ref{thm:piecewiseLinear} and~\ref{thm::r2convex} are illustrated in Fig.~\ref{fig::convexhull} for alphabet sizes  $|\X|=3$ and $ |\Y|=5$, where we considered the TV distance and distortion given by a random matrix $\bfD$. We numerically solve \eqref{eq:DPdiscrete_explicit} for different values of $P$ along the DP tradeoff and project the optimal solutions according to \eqref{eq::projrction0}-\eqref{eq::projrction1}. We also calculate the extreme points of the feasible set to obtain $\mathcal{S}_2$ (for a discussion about finding the vertices of a feasible set, we refer the reader to \cite[Sec.~2.2]{bertsimas1997introduction}). It can be seen that optimal solutions to \eqref{eq:DPdiscrete_explicit} correspond to the linear segments of the DP function, and are obtained on extreme points of $\mathrm{conv}\left(\mathcal{S}_{2}\right)$ in the $(p_0,p_1)$-plane.

\subsection{Full characterization of channels with binary sources}

\label{sec::binarychannel}

We next focus on the case of binary sources, where $\X = \{x_1,x_2\}$ 
with probabilities $p_{x_1},p_{x_2}$, respectively, 
and $\Y$ is of arbitrary size $n_y$. 
It suffices to analyze the TV distance \eqref{eq:tv::def} as the perceptual index, since every metric defining the Wasserstein-$1$ distance is proportional to the Hamming distance in the binary case. The distortion matrix is arbitrary, yielding the matrix $\rho'$ defined in \eqref{eq:rho_prime}.
Denote 
    $u_y = \hlf (\rho'_{\hat{x}_1y}-\rho'_{\hat{x}_2y})$
which is half the cost of reconstructing $y$ as $x_1$ over reconstructing as $x_2$, and assume \WLOG that
    $u_{y_1} \leq u_{y_2}\leq \ldots \leq u_{y_n}$. 
 We define $P_Y^-(u)={\rm Pr}\{u_Y \leq u\}=\sum_{y:u_y\leq u}\bfP_Y(y)$, which is right-continuous with left limit 
$P_Y^-(u^-)={\rm Pr}\{u_Y < u\}=\sum_{y:u_y< u}\bfP_Y(y)$.
We further denote the symbols $y^*_i$ whose $u_y$ is non-zero, namely 
\begin{eqnarray}
0 = u_0 < u_1=u_{y^*_1}\leq \ldots
    \leq u_{M^+} = u_{y^*_{M^+}},
    \\
   u_{-{M^-}} = u_{y^*_{-M^-}}\leq \ldots \leq u_{-1}=u_{y^*_{-1}}< 0 = u_0.
\end{eqnarray}
\begin{theorem}
\label{thm::DP:binarysource}
Assume that $p_{x_1}\geq P_Y^-(0)$, and let $I=\max\{i\colon p_{x_1} \geq P_Y^-(u_i) \}$. Then, the DP function $D(P)$ is piecewise linear with breakpoints $\{P^*_i\}_{i=0}^I$ given by 
\begin{equation}
\label{eq:Pibrakpoint}
P^*_i = p_{x_1}-P_Y^-(u_i)
\end{equation}
where, specifically, $P_0^*=p_{x_1}-P_Y^-(0)=P^*$. 
The DP function is then given by 
\begin{equation}
D(P)=\begin{cases}
D^{*}, & P\geq P_{0}^{*}\\
D(P_{i-1}^{*})+2u_{i}\left(P_{i-1}^{*}-P\right), & P_{i}^{*}\leq P\leq P_{i-1}^{*}\\
D(P_{I}^{*})+2u_{I+1}\left(P_{I}^{*}-P\right), & 0\leq P\leq P_{I}^{*}
\end{cases}.
\end{equation}
If $P_Y^-(0^-) \geq p_{x_1}$, then similarly $P^*_0=P_Y^-(0^-)-p_{x_1}$, and $P^*_i=P_Y^-(u_{-i-1})-p_{x_1}$, while it is non-negative,
and $D(P)$ is determined analogously.
In the case $P_Y^-(0) \geq p_{x_1} \geq  P_Y^-(0^-)$, $P^*=0$ and $D(P)\equiv D^*$ for all $P\geq0$.
\end{theorem}

\begin{remark}
\label{rem:degenerate_u}
If $u_{i}=u_{i-1}$ then $P^*_{i}=P^*_{i-1}$ and this yields a `degenerate' interval.
If $u_{i} > u_{i-1}$, then \eqref{eq:Pibrakpoint} can alternatively be written more simply as $P^*_i=P^*_{i-1}-\bfP_Y(y^*_{i})$.
\end{remark}

The results of Theorem \ref{thm::DP:binarysource} are illustrated in Fig.~\ref{fig:dp}. These results reassure the intuition that channel outputs in $\Y$ should be mapped to symbols in $\{x_1,x_2\}$ in a \textit{greedy} fashion;
At the point $P=1$, each $y$ is reconstructed with a minimal penalty, without any perceptual constraints (as in Proposition~\ref{thm:Dstar}). This can be done by setting, \textit{e.g.}, $q(\hat{x}_1|y) = \delta_{u_y\leq0}$ . At the point $P=P^*$, $y$'s are still reconstructed optimally, but now under a  perception constraint. This can be obtained by rearranging the mapping of symbols whose $u_y=0$, which yields no extra cost in distortion.
Now, suppose that $x_1$ is not `fully allocated', that is, $p_{x_1} \geq P^-_Y(0)$. As the perception constraint becomes more restrictive (lower $P$), the estimator will seek for the minimal cost symbols $y\in \Y$ that are mapped to $x_1$ with probability less than $1$, and increase this probability. For a small change of $\Delta P$, the cost in distortion is  $2u_y \Delta P$. This is done until $P=0$ is met, namely $p_{\hat{x}_1}=p_{x_1}$.

\begin{corollary}
    At the breakpoints  where $P^*_i\neq 0$, an optimal estimator is given by a \textit{deterministic} rule $Q_{P^*_i}$ (for $p_{x_1}\geq P_Y^-(0)$, given by   $Q_{P^*_i}=\left\{q(x_1|y)=\delta_{u_y\leq u_i}\right\}$). Interestingly, at $P \in \left[P_i^*, P_{i-1}^* \right]$, the estimator is given by the \textit{convex} combination of estimators at the interval edges, $Q_P = \alpha Q_{P_{i-1}^*}+(1-\alpha)Q_{ P_{i}^*}$, with $\alpha = \frac{P-P_i^*}{P_{i-1}^* -P_i^*}$.
\end{corollary}

\noindent This result implies that in order to construct an estimator for any point along the tradeoff at test time, without any additional calculations, it is sufficient to calculate $\mathcal{O} (|\Y|)$ estimators beforehand, one at each breakpoint (and at $P=0$).

\section*{ Acknowledgements}
The research of NW was partially  supported by the Israel Science Foundation (ISF), grant no. 1782/22.
The work of RM was partially supported by the Skillman chair in biomedical sciences and by 
the Ollendorff Minerva Center, ECE Faculty, Technion.

\balance
\bibliographystyle{IEEEtran}
\bibliography{dp,dp_chapt}
\balance

\balance
\clearpage

\appendix
In this Appendix, we start with an extended review of Linear Programs and their Dual Problems. We derive the dual forms of both formulations  \eqref{eq:DP::LPform} and \eqref{eq:DP::OTform} (Eq. \eqref{eq:DPdiscrete_explicit}). We then provide a detailed proof for Theorem \ref{thm::DP:binarysource} in the text.

\subsection{The linear optimization problem and strong duality}

Let the general Linear Programming (LP) problem \cite{bertsimas1997introduction}
\begin{equation}
\begin{cases}
\rho\bullet \bfQ & \rightarrow\min_\bfQ\\
\st & a_{i}\bullet \bfQ=b_{i},\,i\in M_{1}~.\\
 & s_{i}\bullet \bfQ\leq b_{i},\,i\in M_{2}\\
 & \bfQ\geq0
\end{cases}
\label{appeq:dp::stndrd_primal}
\end{equation}
$\bfQ,\rho,a_i$ are real $|\mathcal{X}|\times|\mathcal{Y}|$ matrices, $b=\{b_{i}\}_{i\in M}\in\mathbb{R}^{n_{c}}$.
The Dual Linear Programming problem  (DLP) is given by
\begin{equation}
\begin{cases}
w^{\T}b & \rightarrow\max_w\\
\st & w_{i}\leq0,\,i\in M_{2}\\
 & w_{i}\in\mathbb{R},\,i\in M_{1}\\
 & \sum_{i\in M_1}w_{i}\{a_{i}\}_{x,y}+ \sum_{i\in M_2}w_{i}\{s_{i}\}_{x,y}\leq\rho_{x,y},\, \\ & \forall x,y\in\mathcal{X}\times\mathcal{Y}
\end{cases}.
\label{appeq:dp::stndrd}
\end{equation}
Recall that by slight abuse of notation, here, similarly to the main text, we use $x=x_\alpha$ and $y=y_\beta$ to denote their indices $\alpha$ and $\beta$, respectively.

Dual problems are useful for establishing lower bounds on the optimal value, due to the property of \textit{weak duality}, which assures that every feasible value for the Primal problem is greater than or equal to every feasible value of its Dual, yielding (in case where both problems are feasible)
\begin{equation}
\min_{\bfQ}\rho\bullet \bfQ\geq\max_{w}w^{\T}b.
\label{eq::DPdiscrete:weak duality}
\end{equation}

For feasible, bounded LP problems we further possess a \textit{strong duality},
namely the problem \eqref{appeq:dp::stndrd} is feasible and
\begin{equation}
\min_{\bfQ}\rho\bullet \bfQ=\max_{w}w^{\T}b.
\label{appeq::DPdiscrete:strong duality}
\end{equation}

\subsection{A dual form for the TV distance setting}

For our future analysis, here it is convenient to derive the dual of the form  \eqref{eq:DP::LPform} to  $D(P)$. In this formulation, we have $\rho=\bfD^{\T}\bfP_{X,Y}$, and we can write the parameter $b$ in \eqref{appeq:dp::stndrd_primal} as
\begin{gather}
b^{\T}=b(P)^{\T} \\ \nonumber
=\left[p_{y_{1}},\ldots,p_{y_{n}},2P-S_{1}^{\T}\bfP_{X},\ldots,2P-S_{2^{|\mathcal{X}|}-2}^{\T}\bfP_{X}\right],
\end{gather}
where $P$ is the perception index. Also,
\begin{align}
a_{j} &=\bfP_Y({y_{j}})\boldsymbol{1}^{|\mathcal{X}|\T}e_{j},\,j=1,\ldots,|\mathcal{Y}|~,\\
s_{i} &=S_{i}\bfP_{Y}^{\T},\,i=1,\ldots,2^{|\mathcal{X}|}-2~,
\end{align}
$S_{i}$ are the vectors of the set 
 $\left\{-1, 1 \right\}^{|\mathcal{X}|} \backslash \left\{ \pm \left[1,...,1\right] \right\}$, and
 $e_j$ is the $j$-th unit vector in the standard basis.

For convenience, let us split the decision variables in \eqref{appeq:dp::stndrd} into two groups; $|\mathcal{Y}|$ variables $\{w_y\}_{y\in \mathcal{Y}}$ related to the stochasticity constraint for each symbol in $\mathcal{Y}$, and the $2^{|\mathcal{X}|}-2$ variables $\{\nu_i\}$ related to the perception constraints \eqref{eq:dtv-2x-linear-constraints}.
Now, \eqref{appeq:dp::stndrd} becomes
\begin{equation}
\begin{cases}
[w^{\T},\nu^{\T}]b & \rightarrow\max_{w,\nu}\\
\st & \nu_{i} \leq 0,\,\,i=1,\ldots,2^{|\mathcal{X}|}-2\\
 & w_{y}\in\mathbb{R},\,\forall y \in \mathcal{Y} \\
 & \sum_{j}w_{j}\{a_{j}\}_{x,y} - \sum_{i}\nu_{i}\{s_{i}\}_{x,y}\leq\rho_{x,y},\\ &\forall x,y\in\mathcal{X}\times\mathcal{Y},
\end{cases},\label{appeq:DPdiscrete}
\end{equation}
where $\{a_j\}$ is a matrix containing $\bfP_Y({y_j})$ along its $j$-th column, and $0$ elsewhere. Hence, for every $y\in\mathcal{Y}$, only the corresponding $\{a_j\}$ (where $y_j=y$) contributes to the sum, namely $\sum_{j}w_{j}\{a_{j}\}_{x,y} = \bfP_Y(y)w_y $ regardless of $x$. In addition, $s_i$ is a matrix where the $j$-th column is given by $\bfP_Y(y_j)S_i$, thus $\{s_{i}\}_{x,y}=\bfP_Y(y)\{S_i\}_x$. Recall that we denote 
\begin{equation}
\rho'_{\hat x,y}=\frac{\rho{}_{\hat x,y}}{p_Y({y})},
\end{equation}
and directly write \eqref{appeq:DPdiscrete} as
\begin{equation}
\begin{cases}
\sum_{y}\bfP_{Y}(y)w_{y} & +2P\left(\sum_{i}\nu_{i}\right)-\left(\sum_{i}\nu_{i}S_{i}^{\T}\right)\bfP_{X}\rightarrow\max_{w,\nu}\\
\st & \nu_{i}\leq0,\,\,i=1,\ldots,2^{|\mathcal{X}|}-2\\
 & w_{y}\leq\rho'_{x,y}+\sum_{i}\nu_{i}\{S_{i}\}_{x},\,
 \forall x,y\in\mathcal{X}\times\mathcal{Y}
\end{cases}
\label{appeq:DPdiscrete_explicit}
\end{equation}
Note again that $\{S_i\}_x$ is the sign of $\bfP_X(x)-\sum_y \bfQ(x|y)\bfP_Y(y)$ in the $i$-th constraint of \eqref{eq:dtv-2x-linear-constraints}, and that $\sum_i \nu_i S_i^{\T}\bfP_x = \sum_x \bfP_X(x)\sum_i \nu_i \{S_i\}_x $. 
We can write the dual feasible set as
\begin{equation}
\mathcal{S}=\left\{ p_{i}=[w^i, \nu^i]\in\mathbb{R}^{|\Y|+R}:p_{i}^{\T}\bfA\leq c\right\},
\end{equation}
where $R=2^{|\X|}-2$, $\otimes$ is the Kronecker product,
\begin{equation}
\bfA=\left[\begin{array}{c|c}
1^{|\X|}\otimes \mathrm{diag}\{\bfP_{Y}\} & 0_{|\Y|\times R}\\
\hline -S\otimes \bfP_{Y}^{\T} & I_{R}
\end{array}\right]\in\mathbb{R}^{(|\Y|+R)\times(|\X||\Y|+R)},
\end{equation}
and
\begin{equation}
c=\left[\mathrm{rowstack}\{\rho\}^\T,0_{1\times R}\right]\in\mathbb{R}^{|\X||\Y|+R}.
\end{equation}
 Now, $S$ is a matrix whose $j$-th row is $S_{j}^{T}$. Since the Primal problem possesses a finite optimal solution for $P\in[0,\infty)$, its dual must be feasible and bounded (\cite[pp.151]{bertsimas1997introduction}).
 It is then easy to see that $\mathrm{rank}(\bfA)=|\Y|+R$, hence its columns span $\R^{|\Y|+R}$ and the set $\mathcal{S}$ has an extreme point (see \cite[Thm. 2.6]{bertsimas1997introduction}).

\subsection{Derivation of Eq. \eqref{eq:DPdiscrete_explicit}}

Recall the notation $\rho =  \bfD^\T \bfP_{X,Y}$. The program \eqref{eq:DP::OTform} can be written as
\begin{equation}
\min_q c^\T q, \st \, \bfA q = b, q\geq0 
\end{equation}
where $q= \left[ {\rm rowstack}\left\{\bfQ \right\},{\rm rowstack}\left\{\bfPi\right\},\varepsilon \right]$,
\begin{equation}
\bfA = 
\left[
\begin{array}{c|c|c}
1^{1\times|\X|} \otimes {\rm diag}\left\{\bfP_Y\right\} & 0&0 
\\ \hline
0& I_{|\X|} \otimes 1^{1\times|\X|} & 0 
\\ \hline
I_{|\X|} \otimes \bfP_Y^\T & -1^{1\times|\X|} \otimes I_{|\X|} &0 
\\ \hline
0& {\rm rowstack}\left\{H\right\}^\T &1
\end{array} \right]
\end{equation}
and
\begin{eqnarray}
 b =& \left[ \bfP_Y^\T, \bfP_X^\T, 0_{|\X|},P \ \right]^\T, &
\\
c=&\left[\mathrm{rowstack}\{\rho\},0_{1\times (|\X|^2 + 1)}\right]\in\mathbb{R}^{|\X|(|\Y|+|\X|)+1}. &
\end{eqnarray}

We now can write \eqref{eq:dp::stndrd} as 
\begin{equation}
\max_{\tilde{w}} \tilde{w}^\T b, \, \st \  \tilde{w}^\T \bfA \leq c^\T.
\end{equation}
We denote  $\tilde{w}=[w,r,\nu,l]$ where $w \in \R^{|\Y|},r,\nu \in \R^{|\X|}$ and $l$ is a scalar. 
The first $|\X|\times|\Y|$ columns of $\bfA$ yield inequalities of the form 
\begin{equation}
p_y w_y + p_y \nu_{\hat x} \leq \rho_{\hat x ,y}, \ \hat x,y \in \X \times \Y
\end{equation}
while the  next $|\X|^2$ columns yield
\begin{equation}
r_x -  \nu_{\hat x} + \bfH_{x, \hat x} l \leq 0, \ x,\hat x \in \X^2.
\end{equation}
Finally, the last inequality simply says $l\leq0$.

Now, the above is given in the matrix form 
\begin{flalign}
&\max_{w,r,\nu,l} \left[ w^\T\bfP_Y+r^\T \bfP_X + lP\right] &
\\
\nonumber
&\st 
\begin{cases}
    1^{|\X|\times 1}\otimes (\bfP_Y \hadm w)^\T + 1^{1\times|\Y|}\otimes(\bfP_Y^\T \hadm \nu) \leq \rho
    \\
    1^{|\X|\times 1}\otimes r^\T - 1^{1\times|\X|}\otimes \nu + \bfH^\T \cdot l \leq 0
    \\ l \leq 0,
\end{cases}&
\end{flalign}
where $\otimes$ is the Kronecker product and $\hadm$ is the Hadamard (elementwise) product. Inequalities between matrices are applied elementwise. To obtain \eqref{eq:DPdiscrete_explicit}, we replace a variable sign ($l \to -l$) and use the connection \eqref {eq:rho_prime}, $\rho = (1^{|\X|}\otimes \bfP_Y^\T)\hadm \rho'$.

It is easy to see that in this case ${\rm rank}(\bfA) = |\Y| + 2|\X|$ while one constraint is redundant, namely we can eliminate a linear constraint from the primal program (a row of $\bfA$) such that the \textit{row rank} of the problem is full. Equivalently, we can set one of the variables $r_x,\nu_{\hat x}$ to $0$, and the dual feasible set (projected onto $\R^{|\Y|+2|\X|}$) will not contain a line. This implies the existence of an extreme point in the dual feasible set.

\subsection{Full characterization for binary sources (proof of Theorem \ref{thm::DP:binarysource})}
\label{appsec::binarychannel}

Recall we discuss the case of binary sources where $\X = \{x_1,x_2\}$ 
with probabilities $p_{x_1},p_{x_2}$ respectively, 
and $\Y$ is of an arbitrary size $n_y$. 
As a perceptual index, we consider the TV distance \eqref{eq:tv::def} while the distortion matrix is arbitrary, yielding the matrix $\rho'$ defined in \eqref{eq:rho_prime}.
Denote
    $u_y = \hlf \left(\rho'_{\hat{x}_1y}-\rho'_{\hat{x}_2y}\right)$
which is half the cost of reconstructing $y$ as $x_1$ over reconstructing as $x_2$, and we assume \WLOG that
    $u_{y_1} \leq u_{y_2}\leq \ldots \leq u_{y_n}$.
    We define $P_Y^-(u)={\rm Pr}\{u_Y \leq u\}=\sum_{y:u_y\leq u}\bfP_Y(y)$.
We further denote the symbols $y^*_i$ whose $u_y$ is non-zero, namely
\begin{eqnarray}
0 = u_0 < u_1=u_{y^*_1}\leq \ldots
    \leq u_{M^+} = u_{y^*_{M^+}},
    \\
   u_{-{M^-}} = u_{y^*_{-M^-}}\leq \ldots \leq u_{-1}=u_{y^*_{-1}}< 0 = u_0.
\end{eqnarray}

\begin{theorem}
\label{appthm::DP:binarysource} (Theorem \ref{thm::DP:binarysource} in the main text).
Assume that $p_{x_1}\geq P_Y^-(0)$, and let $I=\max\{i\colon p_{x_1} \geq P_Y^-(u_i) \}$. Then, the DP function $D(P)$ is piecewise linear with breakpoints $\{P^*_i\}_{i=0}^I$ given by 
\begin{equation}
\label{appeq:Pibrakpoint}
P^*_i = p_{x_1}-P_Y^-(u_i),
\end{equation}
where, specifically, $P_0^*=p_{x_1}-P_Y^-(0)=P^*$. 
The DP function is then given by 
\begin{equation}
D(P)=\begin{cases}
D^{*}, & P\geq P_{0}^{*}\\
D(P_{i-1}^{*})+2u_{i}\left(P_{i-1}^{*}-P\right), & P_{i}^{*}\leq P\leq P_{i-1}^{*}\\
D(P_{I}^{*})+2u_{I+1}\left(P_{I}^{*}-P\right), & 0\leq P\leq P_{I}^{*}
\end{cases}.
\end{equation}

If $P_Y^-(0^-) \geq p_{x_1}$, then similarly $P^*_0=P_Y^-(0^-)-p_{x_1}$, and $P^*_i=P_Y^-(u_{-i-1})-p_{x_1}$, while it is non-negative,
and $D(P)$ is determined analogously.
In the case $P_Y^-(0) \geq p_{x_1} \geq  P_Y^-(0^-)$, $P^*=0$ and $D(P)\equiv D^*$ for all $P\geq0$.
\end{theorem}

\begin{proof}
Let $\mathcal{X}=\{x_{1},x_{2}\}, \mathcal{Y}=\{y_{1,},\ldots, y_{n_y}\}$. The
dual problem \eqref{appeq:DPdiscrete_explicit} is now written as
\begin{flalign}
\nonumber &\max_{w,\nu} \left[ \sum_{y\in \Y}p_{y}w_{y} \!+\!p_{x_{1}}\left(\nu_{1}\!-\!\nu_{2}\right)\!-\! p_{x_{2}}\left(\nu_{1}\!-\!\nu_{2}\right) \!-\! 2P\left(\nu_{1}\!+\!\nu_{2}\right) \right]& 
\\
&\st \begin{cases}  \ \nu_{1},\nu_{2}\geq 0, &
\\ 
w_{y}\leq\rho'_{\hat{x}_1y}-\left(\nu_{1}-\nu_{2}\right),\ \rho'_{\hat{x}_2y}+\left(\nu_{1}-\nu_{2}\right), \forall y \in \Y
\end{cases}
,
\end{flalign}
where we changed the sign of variables $\nu_i \leftarrow -\nu_i$ for convenience. Now, we denote $u=\nu_{1}-\nu_{2}$. Note that $\nu_1+\nu_2 = |u|+2\min \{\nu_1,\nu_2\}$. Since $P\geq0$ and both $\nu_1,\nu_2$ are non-negative, in an optimal solution we must choose $\min \{\nu_1,\nu_2\}$ to be $0$, which implies $|u|=\nu_1+\nu_2$.  The optimization objective in this case boils down to
\begin{align}
&J(u)= &&\label{eq:DP::n=00003D2}
\\ \nonumber
&\sum_{y\in \Y}p_{y} \min\left\{ \rho'_{\hat{x}_1y}-u,\rho'_{\hat{x}_2y}+u\right\} +\left[2p_{x_{1}}-1\right]u-2P|u|,
&& \end{align}
where in an optimal solution we must have $w_{y}=\min\left\{ \rho'_{\hat{x}_1y}-u,\rho'_{\hat{x}2y}+u\right\} $ since every
$w_{y}$ should be maximal under the constraints. 

We can finally write the dual objective in this case as
\begin{align}
    \nonumber &J_P(u) = \\ &\sum_{y\in\Y}\bfP_Y(y) \min \{ \rho'_{\hat{x}_1y}-u, \rho'_{\hat{x}_2y}+u \} + \left[2p_{x_1}-1\right]u - 2P|u|
    \nonumber
    \\
    \nonumber
     & =\!\!\! \sum_{y:u_y \leq u}\!\!\!\rho_{\hat{x}_1y}\!+\!\!\!\!\sum_{y:u_y > u}\!\!\!\rho_{\hat{x}_2y}\!+\! ( 1\!-\! P_Y^-(u))u\!-\! P_Y^-(u)u\!
     \\ \nonumber
     &+\!(2p_{x_1}\!\!-\!\!1)u \!- \!2P|u|
     \\
    &= \sum_{y:u_y \leq u}\rho_{\hat{x}_1y}+\sum_{y:u_y > u}\rho_{\hat{x}_2y} +2(p_{x_1}- P_Y^-(u))u - 2P|u|. \label{appeq:jpu_fprm}
\end{align}
For any $P\geq0$, this is a \textit{concave} function in the parameter $u$, whose  maximal value is obtained on one of the points where the coefficient of $u$ might change its sign.
\begin{equation}
    D(P)=\max_u J_P(u) = \max \{J_P(0),J_P(u_{y_1}),\ldots, J_P(u_{y_n})\}.
\end{equation}
Note that for each $y\in \Y$, $J_P(u_y)$ is a linear function of $P$, with slope $-2|u_y|$. This is true for $u=0$ as well. The \textit{breakpoints} of $D(P)$ are the points where two (or more) of these functions attain optimality, namely where $\argmax_uJ_P(u)$ contains more than one argument. Since $J_P(u)$ is concave \wrt $u$, $\argmax_uJ_P(u)$ must also contain any interval between these points. 

As we have already seen, for $P \geq 1$ the DP function is flat, 
\begin{equation}
    D(P) = D^* =J_P(0)=\sum_y \min_{\hat{x} }\rho_{\hat{x}y}, \quad P\geq 1.
\end{equation}
It is easy to see from \eqref{appeq:jpu_fprm} that in fact, $J_P(0)=D^*$ for every $P$. To the right of this point,
\begin{equation}
\label{appeq:JP0_rhs}
J_P(u \to 0^+)=D^* + 2(p_{x_1}-P_Y^-
(0))u -2Pu,
\end{equation}
where to the left,
\begin{equation}
\label{appeq:JP0_lhs}
J_P(u\to 0^-)=D^* + 2(p_{x_1}-P_Y^-
(0^-))u +2Pu.
\end{equation}
We comment that $P_Y^-(0^-) = \sum_{y:u_y<0}\bfP_Y(y)\leq P_Y^-(0)$, where $P_Y^-(0^-) = P_Y^-(u_{-1})$ if the latter is defined.

Assume now $p_{x_1}\geq P_Y^-(0) \geq P_Y^-(0^-)$, then for every $P\geq0$, $J_P(u)$ is non-decreasing as $u\rightarrow 0 ^-$. Since it is also concave, the maximal value must be attained at $u=0$ or on the remaining positive candidate points, which we notate  
\begin{equation}
    0< u_1=u_{y^*_1}\leq \ldots
    \leq u_{M^+} = u_{y^*_{M^+}}.
\end{equation}
At the first breakpoint $P^*_0=P^*$, where $0,u_1 \in \argmax J_P(u)$, we should have $J_{P^*}(0)=J_{P^*}(u_1)$, or equivalently
\begin{equation}
    2P_0^*u = 2(p_{x_1}-P_Y^-(0))u, \ 0 \leq u < u_1,
\end{equation}
yielding $P_0^* = p_{x_1}-P_Y^-(0)$.
Similarly, for every possible breakpoint we should have
\begin{equation}
    2P^*_iu = 2(p_{x_1}-P_Y^-(u_i))u, \ u_i \leq u < u_{i+1},
\end{equation}
implying
$P^*_i = p_{x_1}-P_Y^-(u_i)$ for every $i$ such that $p_{x_1} \geq P_Y^-(u_i)$. 
(for a discussion about the case where $u_i$ might be equal to $u_{i+1}$ we refer the reader to Remark~\ref{rem:degenerate_u} in the main text).

If $P_Y^-(0)  \geq  P_Y^-(0^-) \geq p_{x_1}$, then \eqref{appeq:jpu_fprm} is non-decreasing as $u \to 0^+$.
Now, maximum must occur at $u=0$ or the remaining \textit{negative} candidates, $u_{-i}$. By arguments similar to the case above, $P^*_0=P_Y^-(0^-)-p_{x_1}$, and $P^*_i=P_Y^-(u^-_{-i})-p_{x_1}=P_Y^-(u_{-i-1})-p_{x_1}$ while it is non-negative.

Finally, in the case $P_Y^-(0) \geq p_{x_1} \geq  P_Y^-(0^-)$, \eqref{appeq:JP0_lhs} is non-decreasing, while \eqref{appeq:JP0_rhs}  is non-increasing for every $P$, implying $max_u J_P(u) = J_P(0) = D^*$ hence in this case $P^*=0, D(P)\equiv D^*$.
\end{proof}

\end{document}